\documentclass[onecolumn]{IEEEtran}
\IEEEoverridecommandlockouts
\usepackage{cite}
\usepackage{amsmath,amssymb,amsfonts}
\usepackage{amsthm}
\usepackage{algorithm}
\usepackage{algorithmic}
\usepackage{graphicx}
\usepackage{textcomp}
\usepackage{xcolor}
\usepackage{multirow}
\usepackage{pgfplots}

\usetikzlibrary{plotmarks}

\def\BibTeX{{\rm B\kern-.05em{\sc i\kern-.025em b}\kern-.08em
    T\kern-.1667em\lower.7ex\hbox{E}\kern-.125emX}}

\newcommand{\card}[1]{\left|{#1}\right|}

\newcommand{\diff}{d}

\newcommand{\rows}{\mathcal{R}}
\newcommand{\columns}{\mathcal{C}}
\newcommand{\IDMAT}{\mathbf{I}}
\newcommand{\z}{z}
\newcommand{\ex}{\Theta}
\newcommand{\bldb}{\mathbf{b}}
\newcommand{\cc}{\Psi}

\newcommand{\elem}{x}

\newcommand{\SET}{\mathcal{S}}

\newcommand{\NEXT}{e}
\newcommand{\DOMAIN}{\mathcal{X}}
\newcommand{\NDOMAIN}{n}

\newcommand{\PRR}{R}

\newcommand{\HASH}{H}

\newcommand{\SETHASH}{\Xi}

\newcommand{\seed}{\text{seed}}
\newcommand{\CHSUMHASH}{\mathfrak{B}}
\newcommand{\SHASH}{\mathcal{H}}
\newcommand{\NHASHES}{h}
\newcommand{\NFILTER}{N}
\newcommand{\NCHSUM}{C}
\newcommand{\NIBF}{f}

\newcommand{\FILTER}{\mathcal{F}}
\newcommand{\FSMAT}{\mathbf{F}}
\newcommand{\FSMATS}{\mathcal{S}}
\newcommand{\cS}{\mathcal{S}}
\newcommand{\cT}{\mathcal{T}}
\newcommand{\nn}{\mathbb{N}}
\newcommand{\bldzero}{\mathbf{0}}
\newcommand{\Fcell}{\mathtt{c}}
\newcommand{\Fcount}{\mathtt{count}}
\newcommand{\Fsum}{\mathtt{val}}
\newcommand{\Fchsum}{\mathtt{ch}}

\DeclareMathOperator{\Insert}{\mathsf{Insert}}
\DeclareMathOperator{\Remove}{\mathsf{Remove}}
\DeclareMathOperator{\Test}{\mathsf{Test}}
\DeclareMathOperator{\Extract}{\mathsf{Extract}}
\DeclareMathOperator{\Esf}{\mathsf{E}}

\DeclareMathOperator{\Add}{\mathsf{Add}}

\newcommand{\larray}[1]{\langle{#1}\rangle}

\theoremstyle{definition}

\theoremstyle{plain}
\newtheorem{lemma}{Lemma}

\newtheorem{theorem}{Theorem}
\newtheorem{protocol}{Protocol}
\begin{document}

\title{Failure Probability Analysis for Partial Extraction from Invertible Bloom Filters\\
\thanks{This work is supported in part by the grant PRG49 from the Estonian Research Council and by the European Regional Development Fund via CoE project EXCITE.}
}

\author{
	\IEEEauthorblockN{Ivo Kubjas, Vitaly Skachek} \\
	\IEEEauthorblockA{
	\textit{Institute of Computer Science} \\
	University of Tartu, Estonia\\
	Tartu, Estonia \\
	\textit{\{ivo.kubjas,vitaly.skachek\}@ut.ee}
	}
}
\maketitle

\begin{abstract}
Invertible Bloom Filter (IBF) is a data structure, which employs a small set of hash functions. An IBF allows for an efficient insertion and, with high probability, for an efficient extraction of the data. However, the success probability of the extraction depends on the storage overhead of an IBF and the amount of the data stored.
In an application, such as set reconciliation, where there is a need to extract data stored in the IBF, the extraction might succeed only partially, by recovering only part of the stored data.
In this work, the probability of success for a partial extraction of data from an IBF is analyzed. It is shown that partial extraction could be useful in applications, such as set reconciliation. In particular, it allows for set reconciliation by using the IBF, where the storage overhead is too small to allow full extraction. An upper bound on the number of rounds in an iterative set reconciliation protocol is presented. The numerical results are derived analytically, and confirmed by the computer simulations.
\end{abstract}


\section{Introduction}
\label{sec:intro}

\subsection{Background}

Set reconciliation problem~\cite{MTZ03} considers a scenario where two parties $A$ and $B$ possess the sets of data $\SET_A$ and $\SET_B$, respectively,
$\SET_A, \SET_B \subseteq \DOMAIN$. The size of the symmetric difference of the two sets $\diff \triangleq \card{\SET_A \triangle \SET_B}$ is small when compared to the sizes of $\SET_A$ and $\SET_B$. The goal of the problem is to design an efficient protocol, such that after it terminates, both parties possess the set $\SET_A \cup \SET_B$.
The number of parties can also be larger than two.

A na\"ive protocol, which is based on broadcasting the whole sets by each party, is sub-optimal in the cases where $\diff$ is small. Several solutions which achieve communication complexity linear in $\diff$ have been proposed for this scenario: \cite{MTZ03} uses interpolation of characteristic polynomials; \cite{goodrich2011short, goodrich2011invertible} suggest using \emph{Invertible Bloom Filters (IBFs)} for reconciliation;~\cite{EGUV11} specifies such a protocol in full detail.

The IBF is a data structure used to store the set elements (numbers, or, more generally, files). The extraction of the elements from the IBF might fail, depending on the allocation of the storage cells to the elements. This allocation is done according to a fixed small set of hash functions. Typically, for a given set of elements, the larger the size of the IBF, the lower the failure probability is. Therefore, it is an important problem to reduce the failure probability, while at the same time reducing the storage overhead of the IBF.

In this work, we show that even if full extraction of the elements from the IBF fails, the partial extraction is still beneficial
for the reconciliation protocol. Thus, when relying on a partial extraction of the elements, by iterative repetition of (partial) extraction,
one can achieve a full reconciliation efficiently, while using an IBF with a smaller storage overhead. Our analysis is based on
the counting of so-called \emph{state matrices}~\cite{2014IEITF..97.2309Y, Yugawa}, which describe the content of the IBF. The counting allows for estimating the
probability of a failure when at least $\NEXT$ elements are extracted from the IBF.
We compare our results with the existing counterparts in the literature. In particular, we show that by using an iterative protocol, full reconciliation is possible
for the range of parameters, when the counterparts in the literature do not provide for such a result.

This paper is organized as follows. Introduction and notations are presented in Section~\ref{sec:intro}.
IBFs are defined in Section~\ref{sec:IBF}.
The discussion of the setup for partial extraction together with the main lemmas appear in Section~\ref{sec:partial}.
The iterative protocol for set reconciliation is presented and analyzed in Section~\ref{sec:iterative}.
The numerical and simulation results are shown in Section~\ref{sec:numerical}.
The conclusions are stated in Section~\ref{sec:conclusions}.

\subsection{Notations}

In what follows, vectors are denoted by small bold letters, matrices are denoted by capital letters, and scalars are denoted by small regular letters.
We use the notation $\nn$ for the set of natural numbers including zero.

For the set $\cS$, the notation $\card{\cS}$ denotes its cardinality. The subset $\cS_A \triangle \cS_B \triangleq (\cS_A \backslash \cS_B) \cup (\cS_B \backslash \cS_A)$
is a symmetric difference of the sets $\cS_A$ and $\cS_B$.
We also define $[n] \triangleq \{1, 2, \ldots, n \}$ and $[a,b] \triangleq \{a, a+1, \ldots, b \}$ for $a, b \in \nn, \; a \le b$.  Let $\bldzero$ be an all-zeros vector, where the length of the vector is clear from the context.

Let $\NEXT \in \nn$. We use the notation $\cT_\NEXT \subseteq \nn^h$ to denote the set of vectors of length $h$ with the entries from the set $[0, \NEXT]$.
The letter $h$ will be defined in the sequel.

\section{Invertible Bloom Filters}
\label{sec:IBF}
An IBF is a data structure that supports operations $\Insert()$, $\Remove()$, $\Test()$ and $\Extract()$. IBFs are constructed using a set of hash functions $\HASH_i \in \SHASH$, $i \in [\NHASHES]$, where each $\HASH_i$ maps inputs in the domain $\DOMAIN$ to a significantly smaller set $[\NFILTER]$. Hereafter, we assume that all $\HASH_i$ can be described efficiently, and that they can be enumerated. We also assume that the values $\HASH_i(x)$ are distributed uniformly, i.e. for a uniform selection $\elem \in \DOMAIN$ the probability of $\HASH_i(\elem)$ to take any value in $[\NFILTER]$ is exactly $1/\NFILTER$. In addition, we make use of an uniform hash function $\CHSUMHASH$ for checksum, which is also defined over $\DOMAIN$, but its range is $[\NCHSUM]$ for a large constant $\NCHSUM \in \nn$. When constructing an IBF, we choose hash functions $\HASH_i$ such that
\begin{equation}
\forall \elem \in \DOMAIN \; : \;
\HASH_i(\elem) \neq \HASH_j(\elem) \mbox{ for every } i, j \in [\NHASHES], \; i \neq j \; .
\label{eq:no-collision}
\end{equation}
This condition can be easily achieved if the images of the hash functions $\HASH_i$ are all pairwise disjoint subsets of $[\NFILTER]$.

IBF is defined as an array of cells of the form $\larray{\Fcount, \Fsum, \Fchsum}$, where the field $\Fcount$ contains an integer or a finite field element, field $\Fsum$ contains an element in $\DOMAIN$, and the field $\Fchsum$ contains an element in $[\NCHSUM]$. If we denote the $i$-th cell as $\Fcell_i$, then an IBF $\FILTER$ is an array $\larray{\Fcell_1, \ldots \Fcell_{\NFILTER}}$.

In the beginning, an IBF $\FILTER$ is initialized by setting all the fields of all the cells to zero.
In order to insert an element $\elem$ into an IBF $\FILTER$, $\Insert()$ computes the index $j_i = \HASH_{i}(\elem)$ of the corresponding cell, for each $i \in [\NHASHES]$. Then, the cell $\Fcell_{j_i}$ is updated by incrementing the field $\Fcount$ by one, by adding $\elem$ to the field $\Fsum$, and by adding $\CHSUMHASH(\elem)$ to the field $\Fchsum$. The total number of inserted elements into $\FILTER$ is denoted as $f \triangleq \card{\FILTER}$.

Extraction of the elements from an IBF $\FILTER$ works as following. $\Extract()$ iterates over all cells until it finds a cell $\Fcell_j$ with $\Fcount$ field value $\pm 1$. It is verified that the fields of $\Fcell_j$ satisfy $\Fchsum = \CHSUMHASH(\Fsum)$ when $\Fcount = 1$, and $-\Fchsum = \CHSUMHASH(-\Fsum)$
when $\Fcount = -1$. Then, the corresponding value $\Fsum$ is extracted, namely, it is inserted into a special set of extracted elements $\SET_{\FILTER}$, and the corresponding element is removed using $\Remove()$ from $\FILTER$. The latter does the opposite of what $\Insert()$ does. After that, the extraction procedure proceeds analogously with the rest of the cells in $\FILTER$. If no element is extracted while looping over the cells, then $\SET_{\FILTER}$ is returned and the procedure halts.

\section{Partial Extraction}
\label{sec:partial}
\subsection{Problem statement}

The procedure $\Extract()$ may halt when the IBF $\FILTER$ is empty. In that case,
all the elements were extracted from $\FILTER$. However, there is another possibility, that is
only a proper subset $\SET_{\FILTER}$ of the set of elements in $\FILTER$ has been extracted.
We define the \emph{extraction rate}
of a run of $\Extract()$, as follows:
\begin{equation}
R_{\Esf} = \frac{\card{\SET_{\FILTER}}}{\NIBF} \; .
\end{equation}

In~\cite{goodrich2011invertible}, the success probability of $\Extract()$
was considered only for the case $R_{\Esf} = 1$. The main
result of~\cite{goodrich2011invertible} is given in the next theorem.

\begin{theorem}[Theorem~1~\cite{goodrich2011invertible}]
	\label{thm:goodrich_thm}
	Define $c_{\NHASHES}$ as
	\begin{equation}
	c_{\NHASHES}^{-1} = \sup\lbrace
	\alpha: 0 < \alpha < 1; \forall x \in (0,1),
	1 - e^{-\NHASHES \alpha x^{\NHASHES-1}} < x
	\rbrace.
	\end{equation}

	Then, as long as $\NFILTER$ is chosen such that $\NFILTER > (c_{\NHASHES} +
	\varepsilon) \NIBF_0$ for some $\varepsilon > 0$, $\Extract()$ fails with
	probability $O(\NIBF_0^{-\NHASHES + 2})$ whenever $\NIBF \leq \NIBF_0$, where $\NIBF_0 \triangleq \NIBF_0(\NFILTER, \NHASHES)$ is a
	threshold value derived in~\cite{goodrich2011invertible}.
\end{theorem}

{\bf Remark:}
The quantity $\NIBF_0$ in Theorem~\ref{thm:goodrich_thm} is referred to as the \emph{threshold} in the text. The IBFs with $\NIBF \leq \NIBF_0$ and
$\NIBF > \NIBF_0$ are called the \emph{under-threshold IBF} and \emph{over-threshold IBF}, respectively.
\medskip

It follows from Theorem~\ref{thm:goodrich_thm} that if the ratio between the number of available cells
and the number of inserted elements in the IBF is at least $c_{\NHASHES}$, then the extraction rate is $R_{\Esf}=1$ with probability close to 1.
In Table~\ref{tbl:ibf_overhead}, we reproduce the values of $c_{\NHASHES}$ for a varying number of hash functions $\NHASHES$,
as it was computed in~\cite{goodrich2011invertible} using Theorem~\ref{thm:goodrich_thm}.

\renewcommand{\arraystretch}{1.3}
\begin{table}[htbp]
	\caption{Thresholds for IBF overhead versus number of hash functions~\cite{goodrich2011invertible}.}
	\begin{center}
	\begin{tabular}{c | c c c c c}
	  \hline
		$\NHASHES$ & 3 & 4 & 5 & 6 & 7 \\
		\hline
		$c_{\NHASHES}$ & $1.222$ & $1.295$ & $1.425$ & $1.570$ & $1.721$ \\
		\hline
	\end{tabular}
	\end{center}
	\label{tbl:ibf_overhead}
\end{table}

The result in Theorem~\ref{thm:goodrich_thm} is restricted to the case of full extraction. Moreover, the result is suitable for the asymptotic regime, and it is not directly applicable to finite parameters. However, it might be useful to estimate the fraction of extracted elements in the case of partial extraction, i.e. when $R_{\Esf} < 1$. As we show in the sequel, a small number of iterative executions of partial extraction can be sufficient to extract all the elements in the IBF.

To the best of our knowledge, partial extraction from the IBFs was not studied in the existing literature. The technique used in~\cite{goodrich2011invertible} for proving Theorem~\ref{thm:goodrich_thm} is constrained by the threshold ratio $c_h$ between the number of cells and the number of inserted elements. If that ratio is larger than $c_h$, the extraction is full with high probability, and the question of partial extraction becomes redundant. On the other hand, if the ratio is below $c_h$, the application of the techniques in~\cite{goodrich2011invertible} is not straightforward.

\subsection{IBF state matrix representation}

Next, we describe an alternative method to represent an IBF, which is based on a \emph{state matrix}\cite{2014IEITF..97.2309Y, Yugawa}.
Assume that an IBF $\FILTER$ contains $\NIBF$ elements $\elem_1, \elem_2, \ldots, \elem_{\NIBF}$.
The state matrix $\FSMAT$ of an IBF $\FILTER$
is a $\NFILTER \times \NIBF$ binary matrix which contains $\NIBF \cdot \NHASHES$ non-zero entries, where the entry in row $i$ and column $j$, $F_{i,j}=1$, if there exists $\ell \in [\NHASHES]$ such that $h_{\ell}(\elem_j) = i$, and $F_{i,j}=0$ otherwise.
Due to condition~\eqref{eq:no-collision}, the Hamming weight of every column in $\FSMAT$ is exactly $\NHASHES$. For simplicity, we take $\NFILTER_{\HASH} \triangleq \NFILTER/\NHASHES$, and further assume that the range of $\HASH_i$ is $[(i-1) \NFILTER_{\HASH} + 1, i \NFILTER_{\HASH}]$, $i \in [\NHASHES]$. For this choice, the condition~\eqref{eq:no-collision} is fulfilled, and $\FILTER$ can be viewed as if it is partitioned into sub-IBFs $\FILTER_i$, $i \in [\NHASHES]$. The matrix $\FSMAT$ can also be viewed as a matrix of $h$ blocks $\FSMAT_i$ of size $\NFILTER_{\HASH} \times \NIBF$ each, as follows:

\begin{equation}
\FSMAT = \left[\begin{matrix}
\FSMAT_1 \\
\FSMAT_2 \\
\vdots \\
\FSMAT_{\NHASHES}
\end{matrix}\right].\label{eq:submatrices}
\end{equation}

Consider the set $\FSMATS_{\NFILTER_{\HASH}, \NIBF}$ of all $\NFILTER_{\HASH} \times \NIBF$ binary matrices having all their columns of weight one. The number of such matrices is ${\NFILTER_{\HASH}}^{\NIBF}$. The number of possibilities to choose $\NHASHES$ matrices from the set $\FSMATS_{\NFILTER_{\HASH}, \NIBF}$ (with repetitions, the order of choices is important) is ${\NFILTER_{\HASH}}^{\NIBF \NHASHES}$.

If the $i$-th row of $\FSMAT_\ell$ has Hamming weight one, then it corresponds to a value of $\Fcount$ field equal to 1 in the corresponding $i$-th cell of $\FILTER_\ell$. Thus, in that case it is possible to extract a value $\elem_j$ corresponding to the $i$-th row and $j$-th column (of the matrix $\FSMAT_\ell$), where this nonzero  appears. Such an entry $F_{i, j}$ in $\FSMAT_\ell$ is called a pivot.

Consider the execution of the procedure $\Extract()$ applied to the IBF $\FILTER$. The extraction of the element $x_j$ from the $i$-th cell of $\FILTER_\ell$ can be associated with the removal of the $j$-th column in the corresponding $\FSMAT_\ell$. More generally, the extraction of the element $x_j$ from all $\FILTER_\ell$, $\ell \in [\NHASHES]$, can be associated with the removal of the $j$-th column in $\FSMAT$.
This process is repeated iteratively, and it stops when there are no rows of weight one left. As the matrix $\FSMAT$ is partitioned into submatrices, this means that in every submatrix $\FSMAT_i$, $i \in [\NHASHES]$, there is no row of weight one left. The binary matrix which has no rows of weight one is called a \emph{stopping matrix}~\cite{2014IEITF..97.2309Y}. It is shown therein that the number of the $\NFILTER_{\HASH} \times \NIBF$ stopping matrices is given by the following recursive relation:
\begin{equation}
\z(\NFILTER_{\HASH}, \NIBF) =
{\NFILTER_{\HASH}}^{\NIBF} \\
- \sum_{i=1}^{\min(\NFILTER_{\HASH}, \NIBF)} i! \binom{\NFILTER_{\HASH}}{i} \binom{\NIBF}{i} \z(\NFILTER_{\HASH} - i, \NIBF - i) \label{eq:z}
\end{equation}
for $\NFILTER_{\HASH} \ge 1$ and $\NIBF \ge 1$.

\subsection{Counting argument}

We can now state the first result which gives the number of the state matrices $\FSMAT$ in $\underbrace{{\FSMATS_{\NFILTER_{\HASH}, \NIBF}} \times \cdots \times {\FSMATS_{\NFILTER_{\HASH}, \NIBF}}}_{\NHASHES \mbox{ \footnotesize times}}$, which allow for extraction of at least $\NEXT$ elements.

\begin{lemma}
	\label{lemma:sm_lower_bound}
	Let $\FILTER$ be an IBF with $\NHASHES = 1$, which has $\NFILTER$ cells and $\NIBF$ inserted elements.
	Then, the number of state matrices allowing to extract exactly $\NEXT$ elements is:
	\begin{equation}
	\ex_{1}(\NFILTER_{\HASH}, \NIBF, \NEXT) = \binom{\NIBF}{\NEXT} \binom{\NFILTER_{\HASH}}{\NEXT} \cdot \NEXT! \cdot z(\NFILTER_{\HASH} - \NEXT, \NIBF - \NEXT) \; .
	\label{eq:ex1_h1}
	\end{equation}
\end{lemma}
\begin{proof}
	In order to be able to extract exactly $\NEXT$ elements, the matrix $\FSMAT = \FSMAT_1$ should have exactly $\NEXT$ rows of weight one. We denote the set of indices of these rows as $\rows$ and the set of columns where any of these rows has one as $\columns$. Consider the following submatrices of $\FSMAT$:
	\begin{itemize}
		\item
		\label{enum:identity}
		induced by the rows-columns pair $(\rows, \columns)$ -- every row and column has only a single one, thus it is a permutation of the identity matrix $\IDMAT$;
		\item
		\label{enum:restrow}
		induced by the rows-columns pair $(\rows, [\NIBF] \setminus \columns)$ -- the rows indexed by $\rows$ have only a single one in every row, thus this submatrix is a zero matrix;
		\item
		\label{enum:restcolumn}
		induced by the rows-columns pair $([\NFILTER_{\HASH}] \setminus \rows, \columns)$ -- since every column in $\columns$ has a single one, thus it is a zero matrix;
		\item
		\label{enum:stopping}
		induced by rows-columns pair $([\NFILTER_{\HASH}] \setminus \rows, [\NIBF] \setminus \columns)$ -- since it is not possible to extract any further element, it must be a stopping matrix.
	\end{itemize}

	There are $\binom{\NIBF}{\NEXT}$ ways to choose the row indices $\rows$ and $	\binom{\NFILTER_{\HASH}}{\NEXT}$ ways to choose the column indices $\columns$. Additionally, there are $\NEXT!$ permutations of the columns of $\IDMAT$. There are $\z(\NFILTER_{\HASH} - \NEXT, \NIBF - \NEXT)$ ways to choose the stopping matrix in subcase~\ref{enum:stopping}. Thus, the total number of matrices which allow for extracting exactly $\NEXT$ elements is as given in equation~\eqref{eq:ex1_h1}.
	\end{proof}


Let $\columns_i$, $i \in [\NHASHES]$, be the set of columns in $\FSMAT_i$ which contain nonzero entries appearing in rows of weight one.
Consider the vector	$\bldb = (b_1, b_2, \ldots, b_{\NHASHES}) \in \nn^h$.
Denote $\xi(\bldb) \triangleq \sum_{i \in [\NHASHES]} b_i$.

	We generalize Lemma~\ref{lemma:sm_lower_bound} to the case where the number of hash functions is larger than one. The result is summarized in the following lemma.
	\begin{lemma}
	\label{lemma:sm_bound_generalized}
	Let $\FILTER$ be an IBF with $\NHASHES > 1$ hash functions, which has $\NFILTER$ cells and $\NIBF$ inserted elements.
	Then, the number of state matrices allowing to extract $\NEXT$ elements is at least:
	\begin{equation}
	\ex(\NFILTER_{\HASH}, \NIBF, \NHASHES, \NEXT) = \binom{\NIBF}{\NEXT} \cdot \sum\limits_{\bldb \in \cT_\NEXT \;:\; \xi(\bldb) \geq \NEXT} \Bigg( \cc(\NEXT, \bldb)
	\cdot \prod\limits_{i \in [\NHASHES]} \left[ \binom{\NFILTER_{\HASH}}{b_i} \cdot b_i! \cdot \z(\NFILTER_{\HASH}-b_i, \NIBF-b_i) \right]\Bigg),
	\label{eq:bound}
	\end{equation}
	where $\NFILTER_{\HASH} = \frac{\NFILTER}{\NHASHES}$, and the function $\cc$ is defined recursively as
	\begin{equation}
	\cc(\NEXT, \bldb) = \begin{cases}
	1 \hspace{6ex} \text{if } e = 0 \mbox{ and } \bldb = \bldzero \\
  0 \hspace{6ex} \text{if } e = 0 \mbox{ and } \bldb \neq \bldzero \\
	\prod\limits_{i \in [\NHASHES]} \binom{\NEXT}{b_i} - \sum\limits_{j=0}^{\NEXT-1} \left(\binom{\NEXT}{j} \cc(j, \bldb) \right) \quad \text{otherwise}
	\end{cases}.
	\end{equation}
\end{lemma}

\begin{proof}
Consider the partition of $\FSMAT$ into submatrices $\FSMAT_i$, $i \in [\NHASHES]$, as in \eqref{eq:submatrices}.

In order to extract $\NEXT$ elements from $\FILTER$, we choose the corresponding columns of $\FSMAT$ where these $\NEXT$ elements are extracted from. There are $\binom{\NIBF}{\NEXT}$ ways to choose these columns. We denote by $\columns$ the subset that contains them.

For every $\FSMAT_i$, $i \in [\NHASHES]$, it is possible to extract elements from the column subsets $\columns_i \subseteq \columns$ if the sum of the sizes of $\columns_i$, $i \in [\NHASHES]$, is at least $\NEXT$. Denote by $\cc(\NEXT, \bldb)$ the number of ways to choose the subsets of columns $\columns_i$, ${i \in [\NHASHES]}$, where
$\card{\columns_i} = b_i$. For a fixed vector $\bldb$, $\xi(\bldb) \geq \NEXT$, we can apply the same considerations as in Lemma~\ref{lemma:sm_lower_bound}, thus obtaining:
	\begin{equation}
	 \ex(\NFILTER_{\HASH}, \NIBF, \NHASHES, \NEXT) =
	 \binom{\NIBF}{\NEXT}  \cdot \cc(\NEXT, \bldb) \cdot \prod\limits_{i \in [\NHASHES]} \left[ \binom{\NFILTER_{\HASH}}{b_i} \cdot b_i! \cdot \z(\NFILTER_{\HASH}-b_i, \NIBF-b_i) \right].
	\end{equation}


The function $\cc(\NEXT, \bldb)$ counts the number of choices for the subsets of columns $\columns_i$, $i \in [\NHASHES]$, such that
\begin{equation}
\label{eq:subsets}
\cup_{i \in [\NHASHES]} \columns_i = \columns \; ,
\end{equation}
and $\card{\columns_i} = b_i$ for all $i \in [h]$. In particular, we take $\cc(\NEXT, \bldb) = 1$ for the case $\NEXT = 0$ and $\bldb = \bldzero$, and
$\cc(\NEXT, \bldb) = 0$ for the case $\NEXT = 0$ and $\bldb \neq \bldzero$.

Next, we compute the number of possible ways to choose the subsets $\columns_i \subseteq \columns$. This can be done in $\prod_{i \in [\NHASHES]} \binom{\NEXT}{b_i}$ ways. However, we are only interested in choices where~\eqref{eq:subsets} holds with equality. There are $\binom{\NEXT}{j}$ ways to choose a proper subset of $\columns$ of size $j$. For each choice of $j < \NEXT$, there are $\cc(j, \bldb)$ corresponding choices of $\columns_i$. Thus, we obtain
	\begin{equation}
	\label{eq:recursion}
	\cc(\NEXT, \bldb) = \prod\limits_{i \in [\NHASHES]} \binom{\NEXT}{b_i} - \sum\limits_{j=0}^{\NEXT-1} \left(\binom{\NEXT}{j} \cc(j, \bldb) \right).
	\end{equation}
The lemma statement follows immediately.

{\bf Remark:} the right-hand side of expression~\eqref{eq:bound} is a lower bound on the number of state matrices because
after extraction of the first $\NEXT$ elements from $\FILTER$ it is possible that some element $x$ becomes a pivot,
even if $x$ was not a pivot in the beginning of the execution of extraction.
\end{proof}
\medskip

By expanding the recursion~\eqref{eq:recursion} for the values of $\cc(\NEXT, \bldb)$, we observe that the following holds:
\begin{equation}
\begin{split}
\cc(\NEXT, \bldb) & = \sum_{i=1}^{\NEXT} (-1)^{e-i} \binom{e}{i} \prod_{b \in \bldb} \binom{i}{b}
\end{split}.
\end{equation}
We remark that the same expression can also be obtained directly, without using recursion, by using the inclusion-exclusion principle.

\subsection{Success probability for partial extraction}

Let $\FILTER$ be an IBF with $\NHASHES$ hash functions and $\NFILTER = \NHASHES \NFILTER_{\HASH}$ cells, which stores $\NIBF$ elements.
We introduce a random variable $Y_{\NFILTER_{\HASH}, \NHASHES, \NIBF}$, which represents the number of elements extracted from $\FILTER$.

We note that each element is mapped onto one of the $\NFILTER_{\HASH}$ cells by each of the $\NHASHES$ hash functions.
Therefore, the total number of the state matrices describing the state of $\FILTER$ is ${\NFILTER_{\HASH}}^{\NHASHES \NIBF}$.
Assume that each state matrix is chosen uniformly at random. This assumption represents, for example, the case when each element is inserted into any cell with equal probability, independently of other elements, and each hash function is chosen uniformly at random from $\SHASH$.

For any natural value $y$, $y \le \NIBF$, we have:
\begin{align}
\Pr({Y_{\NFILTER_{\HASH}, \NHASHES, \NIBF} \geq y}) \geq
\sum\limits_{\NEXT = y}^{\NIBF}
\frac{\ex(\NFILTER_{\HASH}, \NIBF, \NHASHES, \NEXT)}{{\NFILTER_{\HASH}}^{\NHASHES \NIBF}} \; , \label{eq:pr_extracted}
\end{align}
where $\ex(\NFILTER_{\HASH}, \NIBF, \NHASHES, \NEXT)$ is given in Lemma~\ref{lemma:sm_bound_generalized}.
However, for the case $\NEXT = 0$, the analogous result holds with equality:
\begin{align}
\Pr({Y_{\NFILTER_{\HASH}, \NHASHES, \NIBF} = 0}) = \frac{\z(\NFILTER_{\HASH}, \NIBF)^{\NHASHES}}{{\NFILTER_{\HASH}}^{\NHASHES \NIBF}} \label{eq:pr_noextract} \; .
\end{align}

We state the complementary result to that of Theorem~\ref{thm:goodrich_thm}, which gives the lower bound for the extraction success probability given
an extraction rate.
\begin{theorem}
	\label{thm:extraction_lower_bound}
	Let $\FILTER$ be an IBF with $\NHASHES$ hash functions and $\NFILTER = \NHASHES \NFILTER_{\HASH}$ cells, which stores $\NIBF$ elements. Then, $\FILTER$ fails to achieve extraction rate $R_{\Esf}$ with probability less than or equal to:
	\begin{align}
	1 - \sum\limits_{\NEXT = R_{\Esf} \cdot f }^{\NIBF}
	\frac{\ex(\NFILTER_{\HASH}, \NIBF, \NHASHES, \NEXT)}{{\NFILTER_{\HASH}}^{\NHASHES \NIBF}} \; ,
	\end{align}
	where each state matrix for $\FILTER$ is chosen uniformly at random.
\end{theorem}

\subsection{Experimental results}
\label{sec:expresult}

Next, we compare Theorem~\ref{thm:extraction_lower_bound} with its counterparts in the literature.
Theorem~\ref{thm:goodrich_thm} is applicable only in the case where $\NIBF \leq \NIBF_0(\NFILTER, \NHASHES)$. Theorem~2 in~\cite{2014IEITF..97.2309Y} does not provide a meaningful result when the right-hand side in~\cite[Equation (8)]{2014IEITF..97.2309Y} is larger or equal to $1$.
By contrast, Theorem~\ref{thm:extraction_lower_bound} provides results about partial extractability of the data, in particular for the combinations of parameters where the known methods do not succeed. There are no results about partial extractability in the literature to the best of our knowledge.

In what follows, we compare the result that follows from Theorem~\ref{thm:extraction_lower_bound} with the empirical results. In the experimental study, for every simulation run, we randomly sample $\NIBF$ elements and instantiate an IBF with randomly chosen hash functions $\HASH_1, \ldots, \HASH_{\NHASHES}$. We extract all the elements and compute the extraction rate $R_{\Esf}$. We then count the fraction of experimental runs when $R_{\Esf}$ is greater than the threshold.

The results are presented in Table~\ref{tbl:bound_example}.
We observe that the number of chosen hash functions $\NHASHES$ is critical in decreasing the extraction failure rate. In~\cite{EGUV11}, the recommended number of hash functions is 3-4 in the case when the IBF overhead is sufficiently large to fit all the inserted elements. However, if the number of the inserted elements exceeds a certain threshold, then a smaller number of hash functions yields a lower failure rate both theoretically and experimentally.

For comparison, we also provide the main term in the corresponding upper bound in~\cite{goodrich2011invertible} and the upper bound in~\cite{2014IEITF..97.2309Y}.
For the cases, when one of these results is not applicable, we write ``N/A'' in the corresponding entry in the table.
We note that the result in~\cite{goodrich2011invertible}, however, is obtained for the asymptotic regime, and it contains constants which were not obtained explicitly.
By using numerical examples, we observe that the gap between the result in~\cite{goodrich2011invertible} and the value that accounts for the constants could be very large even for relatively small parameters. Therefore, for small lengths, the numerical results based on~\cite[Theorem 1]{goodrich2011invertible} are not indicative. However, we present them in Table~\ref{tbl:bound_example} for the completeness of the discussion.

\renewcommand{\arraystretch}{1.09}
\begin{table}[htbp]
	\caption{IBF extraction failure probabilities for $\NFILTER=120$.}
	\begin{center}
	  \begin{tabular}{
				| @{\hspace{3pt}} c @{\hspace{3pt}}
				| @{\hspace{3pt}} c @{\hspace{3pt}}
				| @{\hspace{3pt}} c @{\hspace{3pt}}
				| @{\hspace{3pt}} c @{\hspace{3pt}}
				| @{\hspace{0pt}} c	@{\hspace{0pt}}
				| @{\hspace{3pt}} c @{\hspace{3pt}}
				| @{\hspace{3pt}} c @{\hspace{3pt}}
				|}
			\hline
			$\NHASHES$ &
			$\NIBF$ &
			\cite{goodrich2011invertible} &
			\cite{2014IEITF..97.2309Y} &
			\parbox{1.2cm}{\centering $R_{\Esf}$} &
			Theorem~\ref{thm:extraction_lower_bound} &
			Simulation \\
			\hline
			\hline
			
\multirow{4}{*}{2} &
\multirow{4}{*}{60} &
\multirow{4}{*}{$9.75 \cdot 10^{-1}$} &
\multirow{4}{*}{N/A} &
$0.1$ & $3.68 \cdot 10^{-17}$ & $0$ \\
\cline{5-7} &&&&
$0.2$ & $3.12 \cdot 10^{-11}$ & $0$ \\
\cline{5-7} &&&&
$0.5$ & $3.86 \cdot 10^{-2}$ & $0$ \\
\cline{5-7} &&&&
$1$ & $1$ & $5.19 \cdot 10^{-1}$ \\
\hline

\multirow{4}{*}{2} &
\multirow{4}{*}{80} &
\multirow{4}{*}{N/A} &
\multirow{4}{*}{N/A} &
$0.1$ & $2.86 \cdot 10^{-14}$ & $0$ \\
\cline{5-7} &&&&
$0.2$ & $6.56 \cdot 10^{-8}$ & $0$ \\
\cline{5-7} &&&&
$0.5$ & $7.43 \cdot 10^{-1}$ & $3.60 \cdot 10^{-3}$ \\
\cline{5-7} &&&&
$1$ & $1$ & $9.40 \cdot 10^{-1}$ \\
\hline

\multirow{4}{*}{2} &
\multirow{4}{*}{100} &
\multirow{4}{*}{N/A} &
\multirow{4}{*}{N/A} &
$0.1$ & $1.04 \cdot 10^{-10}$ & $0$ \\
\cline{5-7} &&&&
$0.2$ & $1.34 \cdot 10^{-4}$ & $0$ \\
\cline{5-7} &&&&
$0.5$ & $1$ & $2.82 \cdot 10^{-1}$ \\
\cline{5-7} &&&&
$1$ & $1$ & $1$ \\
\hline

\multirow{4}{*}{2} &
\multirow{4}{*}{120} &
\multirow{4}{*}{N/A} &
\multirow{4}{*}{N/A} &
$0.1$ & $3.67 \cdot 10^{-7}$ & $0$ \\
\cline{5-7} &&&&
$0.2$ & $4.36 \cdot 10^{-2}$ & $7.00 \cdot 10^{-4}$ \\
\cline{5-7} &&&&
$0.5$ & $1$ & $9.84 \cdot 10^{-1}$ \\
\cline{5-7} &&&&
$1$ & $1$ & $1$ \\
\hline

\multirow{4}{*}{3} &
\multirow{4}{*}{60} &
\multirow{4}{*}{$1.21 \cdot 10^{-1}$} &
\multirow{4}{*}{$3.17 \cdot 10^{-2}$} &
$0.1$ & $2.88 \cdot 10^{-15}$ & $0$ \\
\cline{5-7} &&&&
$0.2$ & $1.99 \cdot 10^{-9}$ & $0$ \\
\cline{5-7} &&&&
$0.5$ & $2.37 \cdot 10^{-1}$ & $0$ \\
\cline{5-7} &&&&
$1$ & $1$ & $7.90 \cdot 10^{-3}$ \\
\hline

\multirow{4}{*}{3} &
\multirow{4}{*}{80} &
\multirow{4}{*}{$2.17 \cdot 10^{-1}$} &
\multirow{4}{*}{N/A} &
$0.1$ & $6.80 \cdot 10^{-10}$ & $0$ \\
\cline{5-7} &&&&
$0.2$ & $1.65 \cdot 10^{-4}$ & $0$ \\
\cline{5-7} &&&&
$0.5$ & $9.99 \cdot 10^{-1}$ & $1.40 \cdot 10^{-3}$ \\
\cline{5-7} &&&&
$1$ & $1$ & $3.35 \cdot 10^{-2}$ \\
\hline

\multirow{4}{*}{3} &
\multirow{4}{*}{100} &
\multirow{4}{*}{N/A} &
\multirow{4}{*}{N/A} &
$0.1$ & $3.77 \cdot 10^{-5}$ & $0$ \\
\cline{5-7} &&&&
$0.2$ & $1.92 \cdot 10^{-1}$ & $6.00 \cdot 10^{-4}$ \\
\cline{5-7} &&&&
$0.5$ & $1$ & $5.50 \cdot 10^{-1}$ \\
\cline{5-7} &&&&
$1$ & $1$ & $8.73 \cdot 10^{-1}$ \\
\hline

\multirow{4}{*}{3} &
\multirow{4}{*}{120} &
\multirow{4}{*}{N/A} &
\multirow{4}{*}{N/A} &
$0.1$ & $4.34 \cdot 10^{-2}$ & $4.80 \cdot 10^{-3}$ \\
\cline{5-7} &&&&
$0.2$ & $9.80 \cdot 10^{-1}$ & $3.89 \cdot 10^{-1}$ \\
\cline{5-7} &&&&
$0.5$ & $1$ & $1$ \\
\cline{5-7} &&&&
$1$ & $1$ & $1$ \\
\hline

\multirow{4}{*}{4} &
\multirow{4}{*}{60} &
\multirow{4}{*}{$2.22 \cdot 10^{-2}$} &
\multirow{4}{*}{$2.25 \cdot 10^{-3}$} &
$0.1$ & $2.95 \cdot 10^{-3}$ & $0$ \\
\cline{5-7} &&&&
$0.2$ & $2.96 \cdot 10^{-3}$ & $0$ \\
\cline{5-7} &&&&
$0.5$ & $8.21 \cdot 10^{-1}$ & $0$ \\
\cline{5-7} &&&&
$1$ & $1$ & $0$ \\
\hline

\multirow{4}{*}{4} &
\multirow{4}{*}{80} &
\multirow{4}{*}{$3.96 \cdot 10^{-2}$} &
\multirow{4}{*}{N/A} &
$0.1$ & $4.93 \cdot 10^{-3}$ & $0$ \\
\cline{5-7} &&&&
$0.2$ & $9.33 \cdot 10^{-2}$ & $1.00 \cdot 10^{-4}$ \\
\cline{5-7} &&&&
$0.5$ & $1$ & $1.85 \cdot 10^{-2}$ \\
\cline{5-7} &&&&
$1$ & $1$ & $2.48 \cdot 10^{-2}$ \\
\hline

\multirow{4}{*}{4} &
\multirow{4}{*}{100} &
\multirow{4}{*}{N/A} &
\multirow{4}{*}{N/A} &
$0.1$ & $1.03 \cdot 10^{-1}$ & $8.60 \cdot 10^{-3}$ \\
\cline{5-7} &&&&
$0.2$ & $9.83 \cdot 10^{-1}$ & $3.38 \cdot 10^{-1}$ \\
\cline{5-7} &&&&
$0.5$ & $1$ & $9.93 \cdot 10^{-1}$ \\
\cline{5-7} &&&&
$1$ & $1$ & $9.99 \cdot 10^{-1}$ \\
\hline

\multirow{4}{*}{4} &
\multirow{4}{*}{120} &
\multirow{4}{*}{N/A} &
\multirow{4}{*}{N/A} &
$0.1$ & $9.03 \cdot 10^{-1}$ & $6.09 \cdot 10^{-1}$ \\
\cline{5-7} &&&&
$0.2$ & $1$ & $9.98 \cdot 10^{-1}$ \\
\cline{5-7} &&&&
$0.5$ & $1$ & $1$ \\
\cline{5-7} &&&&
$1$ & $1$ & $1$ \\
\hline

\multirow{4}{*}{5} &
\multirow{4}{*}{60} &
\multirow{4}{*}{$5.32 \cdot 10^{-3}$} &
\multirow{4}{*}{$2.64 \cdot 10^{-4}$} &
$0.1$ & $1.56 \cdot 10^{-2}$ & $0$ \\
\cline{5-7} &&&&
$0.2$ & $1.68 \cdot 10^{-2}$ & $0$ \\
\cline{5-7} &&&&
$0.5$ & $9.97 \cdot 10^{-1}$ & $0$ \\
\cline{5-7} &&&&
$1$ & $1$ & $0$ \\
\hline

\multirow{4}{*}{5} &
\multirow{4}{*}{80} &
\multirow{4}{*}{$9.50 \cdot 10^{-3}$} &
\multirow{4}{*}{N/A} &
$0.1$ & $5.05 \cdot 10^{-2}$ & $1.00 \cdot 10^{-3}$ \\
\cline{5-7} &&&&
$0.2$ & $8.24 \cdot 10^{-1}$ & $5.06 \cdot 10^{-2}$ \\
\cline{5-7} &&&&
$0.5$ & $1$ & $4.29 \cdot 10^{-1}$ \\
\cline{5-7} &&&&
$1$ & $1$ & $4.45 \cdot 10^{-1}$ \\
\hline

\multirow{4}{*}{5} &
\multirow{4}{*}{100} &
\multirow{4}{*}{N/A} &
\multirow{4}{*}{N/A} &
$0.1$ & $8.47 \cdot 10^{-1}$ & $4.94 \cdot 10^{-1}$ \\
\cline{5-7} &&&&
$0.2$ & $1$ & $9.83 \cdot 10^{-1}$ \\
\cline{5-7} &&&&
$0.5$ & $1$ & $1$ \\
\cline{5-7} &&&&
$1$ & $1$ & $1$ \\
\hline

\multirow{4}{*}{5} &
\multirow{4}{*}{120} &
\multirow{4}{*}{N/A} &
\multirow{4}{*}{N/A} &
$0.1$ & $1$ & $9.96 \cdot 10^{-1}$ \\
\cline{5-7} &&&&
$0.2$ & $1$ & $1$ \\
\cline{5-7} &&&&
$0.5$ & $1$ & $1$ \\
\cline{5-7} &&&&
$1$ & $1$ & $1$ \\
\hline

		\end{tabular}
	\end{center}
	\label{tbl:bound_example}
\end{table}

\section{Iterative set reconciliation}
\label{sec:iterative}
\subsection{Single-round protocol}

Considers an instance of a set reconciliation problem, where two parties, $A$ and $B$, possess the sets of data $\SET_A$ and $\SET_B$, respectively,
$\SET_A, \SET_B \subseteq \DOMAIN$. Recall that $\diff \triangleq \card{\SET_A \triangle \SET_B}$.

An IBF-based protocol for 2-party set reconciliation uses an addition of two IBFs. The protocol uses a procedure $\Add()$, which takes two IBFs $\FILTER_A$ and $\FILTER_B$ of the same size and returns an IBF $\FILTER$ where every cell value $\Fcount$, $\Fsum$ and $\Fchsum$ is the sum of the corresponding values of the cells in $\FILTER_A$ and $\FILTER_B$. For convenience, we define $-\FILTER$ to be an IBF where all cell values $\Fcount$, $\Fsum$ and $\Fchsum$ are replaced by the additive inverses of themselves. Protocol~\ref{pr:ibf_reconciliation} below is used to reconcile the sets $\SET_A$ and $\SET_B$. In this protocol, we use a pre-selected upper bound $\widetilde{\diff}$ on $\diff$, and the number of hash functions $\NHASHES$. The parties choose $\NFILTER = \left\lceil {c_{\NHASHES} \widetilde{\diff} }\right\rceil$, where $c_{\NHASHES} > 0$ is given in Theorem~\ref{thm:goodrich_thm}.

\begin{protocol}
	\label{pr:ibf_reconciliation}
	\leavevmode
\begin{enumerate}
	\item $A$ and $B$ initialize $\FILTER_A$ and $\FILTER_B$ of size $\NFILTER$, respectively.
	\item For all $\elem \in \SET_A$ and $y \in \SET_B$, do $\FILTER_A \gets \Insert(\FILTER_A, \elem)$ and $\FILTER_B \gets \Insert(\FILTER_B, y)$.
	\item $A$ and $B$ exchange $\FILTER_A$ and $\FILTER_B$.
	\item $A$ and $B$ compute $\FILTER_{\Delta} \gets \Add(\FILTER_A, -\FILTER_B)$.
	\item $A$ and $B$ obtain $\SET_{\Delta} \gets \Extract(\FILTER_{\Delta})$.
	\item $A$ and $B$ add elements from $\SET_{\Delta}$ to the sets $\SET_A$ and $\SET_B$, respectively.
\end{enumerate}
\end{protocol}

Denote $\NDOMAIN = \card{\DOMAIN}$. The communication complexity of the protocol for a fixed $\NHASHES$ is $O(\widetilde{\diff} \cdot \log_2\NDOMAIN)$. If $\diff \approx \widetilde{\diff}$, then Protocol~\ref{pr:ibf_reconciliation} is asymptotically optimal. However, estimating $\diff$ efficiently in the case where $\diff \ll \card{\SET_A \cup \SET_B}$ is non-trivial. Strata Estimator protocol~\cite{EGUV11} uses constant-size IBFs where a increasing subset of the sets are inserted into it. When extraction succeeds, the parties can determine the value of $\diff$ within a factor of two and perform a full reconciliation protocol.

Instead of running the Strata Estimator before Protocol~\ref{pr:ibf_reconciliation}, the parties could instead run it as a single round-trip protocol when the first party sends the Strata Estimator of its set, the second party estimates the symmetric difference size and returns with an IBF of the required size.

\subsection{Iterative reconciliation protocol}

The existing set reconciliation protocols, which are based on IBFs, require sufficiently large overhead. This is because the extraction procedure needs to return all the elements in the symmetric difference~\cite{EGUV11, goodrich2011invertible, MP18}. However, this approach is not always optimal.
For example, in a broadcast network with a single transmitter and several receivers, one-way reconciliation can be used to synchronize the receivers' databases with that of the transmitter. This resembles special cases of the coded caching and index coding~\cite{MAN14, LOJ19, BYBJK11}, where all the missing elements have to be delivered to each of the receivers. Another example is a two-party set reconciliation protocol operating over bandwidth-constrained channel. Thus, the parties might want to run the first round of the protocol quickly to reduce the difference between the two databases.

In the sequel, we define the following protocol where the overhead of an IBF is not sufficiently large for extracting all the elements in the symmetric difference in a single round. However, running the protocol for several rounds allows to reconcile all the elements in the symmetric difference with high probability.

In order to determine if the parties have already reconciled the sets, we define a set-hash function $\SETHASH(\cdot)$ which takes a subset $\SET \subseteq \DOMAIN$ as an input, and it returns a short hash value of this subset. Additionally, we assume that the hash functions $\HASH_1, \HASH_2, \ldots, \HASH_{\NHASHES}$ can be randomly sampled. This can be achieved by defining $\HASH_i(\elem) = \HASH(\seed, i, \elem)$, $i \in [h]$, for some initialization value $\seed$ and well-defined hash function family $\SHASH$. The complete description of the hash function is then given as $(\seed, i)$.

The protocol for iterative reconciliation is presented as Protocol~\ref{pr:iterative_reconciliation}. Here $\NFILTER$ is the size of the IBFs, and $\NHASHES$ is the number of hash functions, whose values are set prior to the execution of the protocol. The protocol terminates when the sets of elements are reconciled, as it is shown in Lemma~\ref{lemma:protocol_termination}.

\begin{protocol}
	\label{pr:iterative_reconciliation}
	\leavevmode
\begin{enumerate}
	\item\label{pr:start} $A$ sends $\SETHASH(\SET_A)$ to $B$.
	\item\label{pr:terminate} $B$ terminates protocol if $\SETHASH(\SET_B) = \SETHASH(\SET_A)$.
	\item $A$ initializes hash functions $\HASH_1, \HASH_2, \ldots, \HASH_{\NHASHES}$ and IBF $\FILTER_A$ of size $\NFILTER$, and inserts all $\elem \in \SET_A$ into $\FILTER_A$.
	\item $A$ sends $\HASH_1, \HASH_2, \ldots, \HASH_{\NHASHES}$, and $\FILTER_A$ to $B$.
	\item $B$ initializes $\FILTER_B$ of size $\NFILTER$, and inserts all $\elem \in \SET_B$ into $\FILTER_B$.
	\item\label{pr:set_diff} $B$ computes $\FILTER_{\Delta} \gets \Add(\FILTER_A, -\FILTER_B)$.
	\item\label{pr:extracted} $B$ computes $\SET_B \gets \SET_B \cup \SET_{\Delta}$, where $\SET_{\Delta} \gets \Extract(\FILTER_{\Delta})$.
	\item\label{pr:b_reply} $B$ sends $\SET_{\Delta'} \gets \SET_{\Delta} \setminus \SET_B$ to $A$.
	\item\label{pr:a_recon} $A$ computes $\SET_A \gets \SET_A \cup \SET_{\Delta'}$.
	\item Go to Step~\ref{pr:start}.
\end{enumerate}
\end{protocol}

\begin{lemma}
	\label{lemma:protocol_termination}
	For any initial sets $\SET_A$ and $\SET_B$, with high probability Protocol~\ref{pr:iterative_reconciliation} terminates with both $A$ and $B$ possessing $\SET_A \cup \SET_B$ if $\NFILTER_{\NHASHES} > 1$.
\end{lemma}

\begin{proof}
	The number of elements reconciled in every rounds correspond to the number of extracted elements from the IBF in Step~\ref{pr:extracted}. In Step~\ref{pr:set_diff}, common elements $\SET_A \cap \SET_B$ cancel out in $\FILTER_{\Delta}$, and thus $\FILTER_{\Delta}$ contains only the elements in $\SET_A \triangle \SET_B$. Thus, in Step~\ref{pr:extracted} the extracted elements $\SET_{\Delta}$ are a subset of $\SET_A \triangle \SET_B$.

	We observe that the protocol terminates. If $\card{\SET_A \triangle \SET_B} > 0$, then, with a nonzero probability (over selections of hash functions and element sets), there is at least one assignment of hash functions $\HASH_{1}, \ldots, \HASH_{\NHASHES}$ which allows for extracting of at least one element. Then,	a number of extracted elements is larger than zero with a non-zero probability.
	Thus, after a sufficiently large number of rounds, the protocol eventually terminates with $\SET_A = \SET_B$ indicated by set-hash equality in Step~\ref{pr:terminate}.
\end{proof}

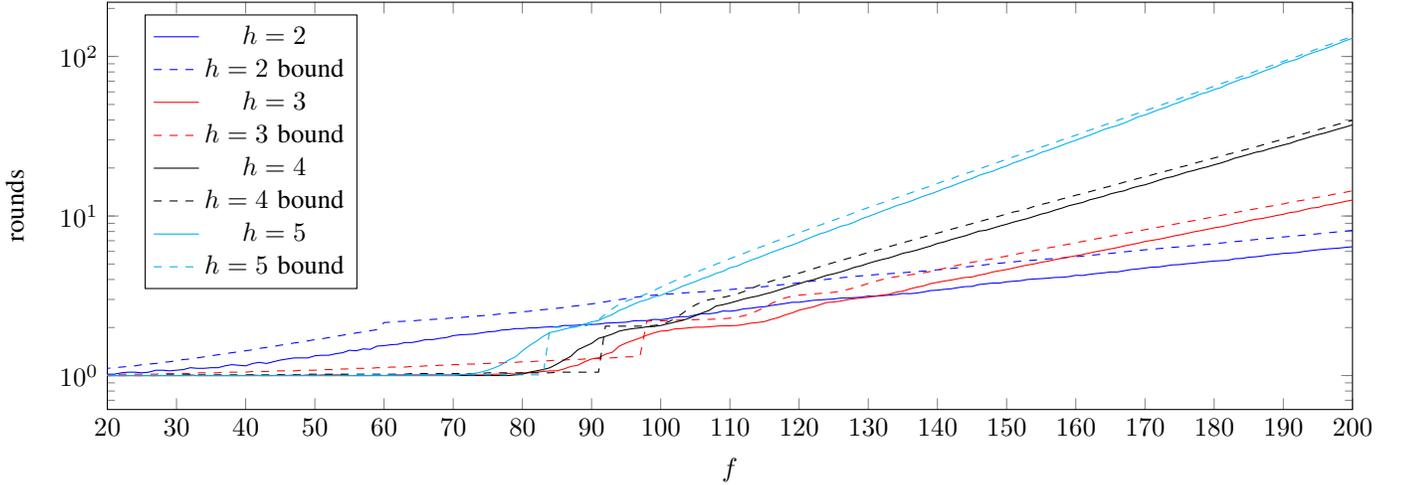
\begin{figure*}[tb]
	\begin{tikzpicture}
	\begin{axis}[
	xlabel={$\NIBF$},
	ylabel={rounds},
	xmin=20,
	xmax=200,
	ymode=log,
	legend pos=north west,
	width=\textwidth,
	height=7cm,
	]
	\addplot[
	blue,
	] table [x=f, y=2, col sep=semicolon] {iterative_recon2.csv};
	\addlegendentry{$h=2$};
	\addplot[
	blue,
	dashed,
	] table[x=x, y=2, col sep=comma]{iterative_recon_data2.csv};
	\addlegendentry{$h=2$ bound}
	\addplot[
	red,
	] table [x=f, y=3, col sep=semicolon] {iterative_recon2.csv};
	\addlegendentry{$h=3$};
	\addplot[
	red,
	dashed,
	] table[x=x, y=3, col sep=comma]{iterative_recon_data2.csv};
	\addlegendentry{$h=3$ bound}
	\addplot[
	black,
	] table [x=f, y=4, col sep=semicolon] {iterative_recon2.csv};
	\addlegendentry{$h=4$};
	\addplot[
	black,
	dashed,
	] table[x=x, y=4, col sep=comma]{iterative_recon_data2.csv};
	\addlegendentry{$h=4$ bound}
	\addplot[
	cyan,
	] table [x=f, y=5, col sep=semicolon] {iterative_recon2.csv};
	\addlegendentry{$h=5$};
	\addplot[
	cyan,
	dashed,
	] table[x=x, y=5, col sep=comma]{iterative_recon_data2.csv};
	\addlegendentry{$h=5$ bound}
	\end{axis}
	\end{tikzpicture}
	\caption{Experimental number of rounds for iterative reconciliation.}
	\label{fig:iterative_rounds}
\end{figure*}

There are several variations of the protocol, which do not change the analysis, but allow it to be efficiently used in different scenarios.
Due to asymmetry of the messages sent in this protocol (set-hash is sent only by $A$, while $B$ replies only with the elements in the set difference), this protocol is suitable for client-server model. One can change the protocol to suit the peer-to-peer model. In order to do this, the exchanged messages should be similar. This can be achieved by omitting Steps \ref{pr:b_reply}--\ref{pr:a_recon} and changing the roles of parties $A$ and $B$ when looping to Step \ref{pr:start}.

It is also possible to turn the protocol into a one-way reconciliation protocol as described above (i.e. only one party has reconciled sets at the end of the protocol). This can be achieved by omitting Steps \ref{pr:b_reply}--\ref{pr:a_recon}.

\section{Numerical results}
\label{sec:numerical}

During the execution of Protocol~\ref{pr:iterative_reconciliation}, the parties need to know when the reconciliation has been completed. It is beneficial, therefore, to know the expected number of the required protocol rounds. Next, we estimate this number analytically, and compare it with the simulation results.

Define the sequence of the random variables $\PRR(i)$, $i \geq 0$, whose values denote the number of the elements in the symmetric difference $\SET_A \triangle \SET_B$ after completion of round $i$ in Protocol~\ref{pr:iterative_reconciliation}. In the sequel, we call the variable $\PRR(i)$ \emph{the $i$-th state} of the protocol.
We remark that the transition from the state $\PRR(i)$ into the next state $\PRR(i+1)$ does not depend on the states $\PRR(i')$ for $i' < i$,
thus forming a Markov chain. It follows from Lemma~\ref{lemma:protocol_termination} that for any its realization, the sequence $\PRR(i)$ is monotonically non-increasing with $i$, and it approaches $0$ for $i \rightarrow \infty$.

Next, we use Lemmas~\ref{lemma:sm_lower_bound}-\ref{lemma:sm_bound_generalized} and Theorem~\ref{thm:goodrich_thm}
for estimation of the transition probabilities between different states. If $\PRR(i) \geq \NFILTER/c_{\NHASHES}$,
then from Lemmas~\ref{lemma:sm_lower_bound}-\ref{lemma:sm_bound_generalized}, we have:
\begin{align}
\label{eq:Markov-1}
\Pr(\PRR(i+1) = \NIBF-\NEXT \mid \PRR(i) = \NIBF) = \frac{\ex(\NFILTER_{\HASH}, \NIBF, \NHASHES, \NEXT)}{\NFILTER_{\HASH}^{\NHASHES \NIBF}} \; .
\end{align}
For $\PRR(i) < \frac{\NFILTER}{c_{\NHASHES}}$, it is observed in \cite{goodrich2011invertible} and \cite{2014IEITF..97.2309Y} that the extraction failure probability is dominated by the case where two different elements are inserted into the same subset of cells. As it is shown in~\cite{goodrich2011invertible}, this yields the following expressions for the probability that two elements are left in the IBF after extraction:
\begin{align}
\Pr(\PRR(i+1) = 2 \mid \PRR(i)) = \binom{\NIBF}{2} \binom{\NFILTER}{\NHASHES} \left(\frac{\NHASHES}{\NFILTER}\right)^{2 \NHASHES} \; ,
\label{eq:Markov-2}
\end{align}
and for the probability that all the elements are successfully extracted (when ignoring constants in the probability expression):
\begin{align}
\Pr(\PRR(i+1) = 0 \mid \PRR(i)) = 1-\binom{\NIBF}{2} \binom{\NFILTER}{\NHASHES} \left(\frac{\NHASHES}{\NFILTER}\right)^{2 \NHASHES} \; .
\label{eq:Markov-3}
\end{align}

We observe that the results of Lemma~\ref{lemma:sm_lower_bound} and Lemma~\ref{lemma:sm_bound_generalized} yield a lower bound on the number of extracted elements. Thus, the expected number of steps obtained from the relations~\eqref{eq:Markov-1}, \eqref{eq:Markov-2} and~\eqref{eq:Markov-3} imply an upper bound on the expected number of rounds in Protocol~\ref{pr:iterative_reconciliation}.

We compute the expected number of steps for $\NFILTER=120$, $\NHASHES \in \{2,3,4,5\}$ and $\NIBF \in [20, 200]$. The results are shown in Figure~\ref{fig:iterative_rounds} by using dashed lines. For $\NHASHES = 2$, we used $c_2 = 2$. For $\NHASHES > 2$, the values of $c_\NHASHES$ are as in Table~\ref{tbl:ibf_overhead}.

We compare the numerical bounds with the simulation results, for the same choices of $\NFILTER$, $\NHASHES$ and $\NIBF \in [20, 200]$. For every set of parameters, we ran the protocol 1000 times with randomly chosen elements and hash functions. The elements are chosen from $\DOMAIN$
 uniformly at random, one by one, while ensuring that the same element is not chosen twice.
The domain $\DOMAIN$ is the field of 256-bit long integers modulo the prime number
$11579208923731619542357098500868790785326998466564$\\
$0564039457584007913129640233$.

For the hash functions, we use the 256-bit long version of SHA-2 with uniformly chosen 32-bit long random seeds~\cite{Sha-2}.
The received value is converted into an integer, and its residue modulo $\NFILTER_H$ is used as the index of the cell in the subfilter.
The average number of rounds for full reconciliation is shown in Fig~\ref{fig:iterative_rounds} by using the solid lines.

We observe that the analytical estimates are quite close to the simulated results. We also observe that, for the selected parameters, the number of rounds is smaller for $\NHASHES \in \{3,4\}$ than for $\NHASHES = 5$ in both under- and over-threshold IBFs. By comparing the results for $\NHASHES = 3$ and $\NHASHES = 4$, we observe that the performance is similar for the under-threshold IBF, but the choice $\NHASHES = 3$ allows for a smaller number of rounds for over-threshold IBF. Since the threshold for $\NHASHES = 3$ is larger than for $\NHASHES = 4$, then a larger number of elements can be reconciled in a single round. We also observe that for $\NHASHES = 2$, the protocol underperforms when IBF is under-threshold. In this case the performance is weaker than for $\NHASHES = 3$ if $\NIBF \leq \NFILTER$, yet it allows for a smaller number of rounds in the case where $\NIBF > \NFILTER$.

\section{Conclusions}
\label{sec:conclusions}

In this work, we presented analysis of failure probability for partial extraction of elements from an IBFs.
The estimates on the failure probability in this work generalized and improved the results in the preceding works, in
particular in~\cite{goodrich2011invertible, 2014IEITF..97.2309Y}.
We also proposed a multi-round protocol for set reconciliation between different parties, which is based on partial
extraction of the data from the IBFs. The IBFs in this protocol required a smaller overhead than their counterparts
in the literature. We analyzed the number of rounds for such a protocol (a) by a recursive analytical formula;
(b) by a simulation. We observed that the analytical result match the simulation results quite closely.

\end{document}